\documentclass[envcountsame,envcountsect,orivec,11pt]{llncs}
\usepackage{\jobname}

\title{The Complexity of Reasoning for Fragments of Default
  Logic\thanks{A preliminary version of this paper appeared in the
    proceedings of the conference SAT'09 \cite{BMTV09a}.
    This work was supported by the German Research Foundation (Deutsche
    For\-schungs\-gemeinschaft) under grants KO 1053/5-2 and VO 630/6-1. 
}}

\author{%
  Olaf~Beyersdorff \and%
  Arne~Meier \and%
  Michael~Thomas \and%
  Heribert~Vollmer}

\institute{%
  Institut f\"ur Theoretische~Informatik, Gottfried Wilhelm Leibniz Universit\"{a}t\newline%
  Appelstr.~4, 30167~Hannover, Germany\newline%
  \protect\url|{beyersdorff,meier,thomas,vollmer}@thi.uni-hannover.de|%
}

\begin{document}

\maketitle

\begin{abstract}
  Default logic was introduced by Reiter in 1980. In 1992, Gottlob
  classified the complexity of the extension existence problem for
  propositional default logic as $\SigmaPtwo$-complete, and the
  complexity of the credulous and skeptical reasoning problem as
  $\SigmaPtwo$-complete, resp.\  $\PiPtwo$-complete. Additionally, he
  investigated restrictions on the default rules, \emph{i.e.}, semi-normal
  default rules. Selman made in 1992 a similar approach with
  disjunction-free and unary default rules. In this paper we 
  systematically restrict the set of allowed propositional connectives.
  We give a complete complexity classification for all sets of Boolean 
  functions in the meaning of Post's lattice for all three common 
  decision problems for propositional default logic. We show that the
  complexity is a hexachotomy ($\SigmaPtwo$-, $\DeltaPtwo$-, $\NP$-, 
  $\P$-, $\NL$-complete, trivial) for the extension existence problem, 
  while for the credulous and skeptical reasoning problem we obtain 
  similar classifications without trivial cases.
  \medskip
  
  \keywordname\ computational complexity, default logic, nonmonotonic reasoning, Post's lattice
\end{abstract}

\section{Introduction}

When formal specifications are to be verified against real-world situations, 
one has to overcome the \emph{qualification problem} that denotes the impossibility 
of listing all conditions required to decide compliance with the specification.
To overcome this problem, McCarthy proposed the introduction of ``common-sense'' 
into formal logic \cite{mccarthy:80}.
Among the formalisms developed since then, Rei\-ter's default logic is one of the 
best known and most successful formalisms for modeling common-sense reasoning. 
Default logic extends the usual logical (first-order or propositional) derivations by patterns
for default assumptions. These are of the form ``in the absence of
contrary information, assume \dots''. Reiter argued that his logic is
an adequate formalization of the human reasoning under the
\emph{closed world assumption}. In fact, today default logic is used in various areas of 
artificial intelligence and computational logic. 

What makes default logic computationally presumably harder than
propositional or first-order logic is the fact that the semantics
(\emph{i.e.}, the set of consequences) of a given set of premises is
defined in terms of a fixed-point equation. The different fixed points
(known as \emph{extensions} or \emph{expansions}) correspond to
different possible sets of knowledge of an agent, based on the given
premises. 

In a seminal paper from 1992, Georg Gottlob classified the complexity of three important decision problems for default logic:
\begin{enumerate}
 \item Given a set of premises, decide whether it has an extension at all.
 \item Given a set of premises and a formula, decide whether the
   formula occurs in at least one extension (so called \emph{brave} or
   \emph{credulous reasoning}). 
 \item Given a set of premises and a formula, decide whether the
   formula occurs in all extensions (\emph{cautious} or
   \emph{skeptical reasoning}). 
\end{enumerate}
While in the case of first-order default logic, all these
computational tasks are undecidable, Gottlob proved that for
\emph{propositional default logic}, the first and second are complete
for the class $\SigmaPtwo$, the second level of the polynomial
hierarchy (Meyer-Stockmeyer hierarchy), while the third is complete for
the class $\PiPtwo$ (the class of complements of $\SigmaPtwo$ sets). 
 
In the past, various semantic and syntactic restrictions have been
proposed in order to identify computationally easier or even tractable
fragments (see, \emph{e.g.},
\cite{stillman:90,kautz-selman:91,ben-eliyahu-zohary:2002}).
 This is the
starting point of the present paper. We propose a systematic study of
fragments of default logic defined by restricting the set of allowed
propositional connectives. For instance, if we look at the fragment
where we forbid negation and the constant $\false$ and allow only conjunction and disjunction,
we show that while the first problem is trivial (there always is an extension, in fact a
unique one), the second and third problem become $\co\NP$-complete. In
this paper we look at all possible sets $B$ of propositional
connectives and study the three decision problems defined by Gottlob
when all involved formulae contain only connectives from $B$. The
computational complexity of the problems then, of course, becomes a
function of $B$. We will see that \emph{Post's lattice} of all closed
classes of Boolean functions is the right way to study all such sets
$B$. Depending on the location of $B$ in this lattice, we completely
classify the complexity of all three reasoning tasks, see
Figs.~\ref{fig:extension-existence} and
\ref{fig:credulous-skeptical-reasoning}. We will show that, depending
on the set $B$ of occurring connectives, the problem of determining the
existence of an extension is either $\SigmaPtwo$-complete,
$\DeltaPtwo$-complete, $\NP$-complete, $\P$-complete, $\NL$-complete, or trivial, 
while for the reasoning problems the trivial cases split up into 
$\co\NP$-complete, $\P$-complete, and $\NL$-complete ones
(under constant-depth reductions).
 
The motivation behind our approach lies in the hope that identifying
fragments of default logic with simpler reasoning procedures may 
help us to understand the sources of hardness for the full
problem and to locate the boundary between hard and easy
fragments. In particular, these procedures may lead to algorithms
for solving the studied problems more efficiently.

  \begin{figure}[ht]
    \centering
    \includegraphics[height=0.62\textheight]{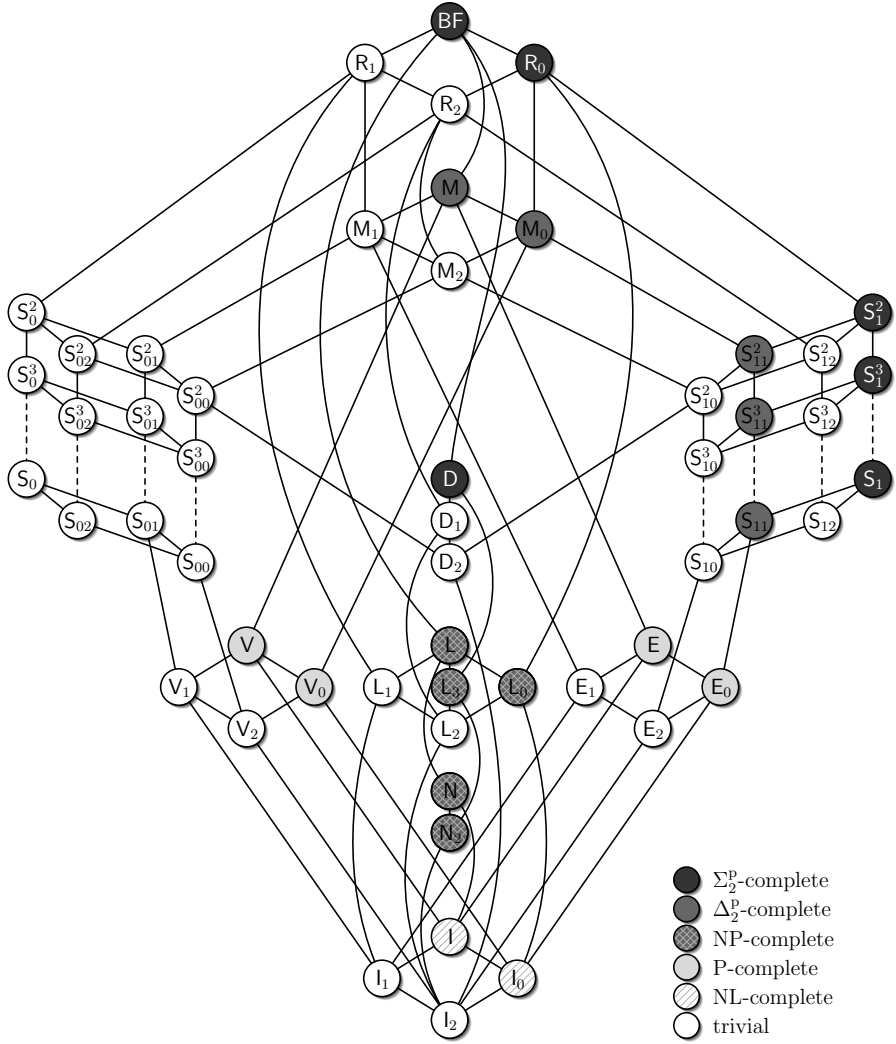}
    \caption{\label{fig:extension-existence}%
    Post's lattice. Colours indicate the complexity of $\EXT(B)$, the Extension Existence Problem for $B$-formulae.}
  \end{figure}

  \begin{figure}[ht]
    \centering
    \includegraphics[height=0.62\textheight]{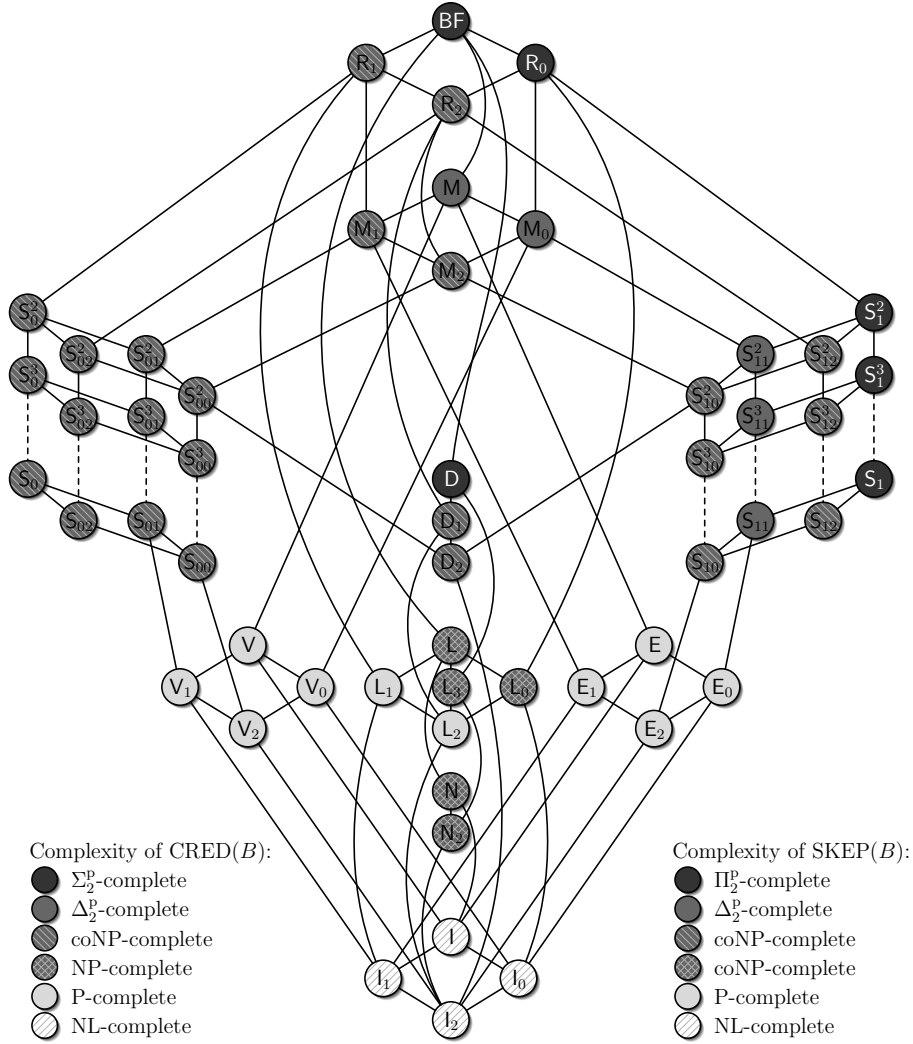}
    \caption{\label{fig:credulous-skeptical-reasoning}%
    Post's lattice. Colours indicate the complexity of $\CRED(B)$ and $\SKEP(B)$, the Credulous and Skeptical Reasoning Problems for $B$-formulae.}
  \end{figure}

  This paper is organized as follows.
  After some preliminary remarks in Section~\ref{sect:prelim}, we
  introduce Boolean clones in Section~\ref{sect:implication}. At this
  place we also provide a full classification of the
  complexity of logical implications for fragments of
  propositional logic, as this classification will serve as a central tool
  for subsequent sections. 
  In Section~\ref{sect:default_logic}, we start to investigate 
  propositional default logic.
  Section~\ref{sect:default_reasoning} then presents our main results on
  the complexity of the decision problems for default logic. 
  Finally, in Section~\ref{sec:conclusion} we conclude with a summary and a discussion.

\section{Preliminaries} \label{sect:prelim}
  
  In this paper we make use of standard notions of complexity
  theory. The arising complexity degrees encompass the classes $\NL$,
  $\P$, $\NP$, $\co\NP$, $\SigmaPtwo$ and $\PiPtwo$.
  For a thorough introduction into the field, the reader is referred to \cite{pap94}.
  For the hardness results, we use \emph{constant-depth} reductions,
  defined as follows:
  A language $A$ is \emph{constant-depth reducible} to a language
  $B$ ($A\leqcd B$) if there exists a logtime-uniform $\AC{0}$-circuit family
  $\{C_n\}_{n\geq 0}$ with unbounded fan-in $\{\land,\lor,\neg\}$-gates and oracle gates
  for $B$ such that for all $x$, $C_{|x|}(x) = 1$ if and only if $x \in A$ (cf.~\cite{vol99}).

  We assume familiarity with propositional logic. 
  The set of all propositional formulae is denoted by $\allFormulae$.
  For $A\subseteq \allFormulae$ and  $\varphi\in \allFormulae$, we
  write $A \models \varphi$ if and only if all assignments satisfying all
  formulae in $A$ also satisfy $\varphi$. 
  By $\theorems{A}$ we denote the set of all consequences of
  $A$, \emph{i.e.}, $\theorems{A}=\{ \varphi \mid A \models \varphi \}$. 
  For a literal $\ell$ and a variable $x$, we define $\sneg{\ell}$ as the literal of opposite polarity, 
  \emph{i.e.}, $\sneg{\ell}:=x$ if $\ell=\neg x$ and $\sneg{\ell}:=\neg x$ if $\ell=x$. 
  For a formula $\varphi$, let $\varphi_{[\alpha / \beta]}$ denote
  $\varphi$ with all occurrences of the formula $\alpha$ replaced by the formula $\beta$, and
  let $A_{[\alpha / \beta]} := \{\varphi_{[\alpha / \beta]} \mid
  \varphi \in A\}$ for $A\subseteq\allFormulae$.

\section{Boolean Clones and the Complexity of the Implication Problem}
\label{sect:implication}

  A propositional formula using only connectives from a finite set $B$
  of Boolean functions is called a $B$-formula. 
  The set of all $B$-formulae is denoted by $\allFormulae(B)$.
  In order to cope with the infinitely many finite sets $B$ of Boolean functions, we require 
  some algebraic tools to classify the complexity of the infinitely many arising reasoning problems. 
  A \emph{clone} is a set $B$ of Boolean functions that is closed under superposition, 
  \emph{i.e.}, $B$ contains all projections and is closed under arbitrary
  composition \cite[Chapter~1]{pip97b}.
  For an arbitrary set $B$ of Boolean functions, 
  we denote by $[B]$ the smallest clone containing $B$ and call $B$ a \emph{base} for $[B]$.
  In \cite{pos41} Post classified the lattice of all clones and found a finite base for each clone, 
  see Fig.~\ref{fig:extension-existence}.
  In order to introduce the clones relevant to this paper, we define the following notions
  for $n$-ary Boolean functions $f$:  
    \begin{itemize} \itemsep 0pt 
      \item $f$ is \emph{$c$-reproducing} if $f(c, \ldots , c) = c$, $c \in \{\false,\true\}$.
      \item $f$ is \emph{monotone} if $a_1 \leq b_1, a_2 \leq b_2, \ldots , a_n \leq b_n$ implies $f(a_1, \ldots , a_n) \leq f(b_1, \ldots , b_n)$.
      \item $f$ is \emph{$c$-separating} if there exists an $i \in \{1, \ldots , n\}$ such that $f(a_1, \ldots , a_n) = c$ implies $a_i = c$, $c \in \{\false,\true\}$.
      \item $f$ is \emph{self-dual} if $f \equiv \dual{f}$, where $\dual{f}(x_1, \ldots , x_n) = \neg f(\neg x_1, \ldots , \neg x_n)$.
      \item $f$ is \emph{linear} if $f \equiv x_1 \xor \cdots \xor x_n
        \xor c$ for a constant $c \in \{0, 1\}$ and variables $x_1,
        \ldots , x_n$. 
    \end{itemize}
    The clones relevant to this paper are listed in Table~\ref{tab:clones}. 
    The definition of all Boolean clones can be found, \emph{e.g.}, in \cite{bcrv03}. 
    
    \begin{table}[tb]
      \centering
      
      \begin{tabular}{c|@{\quad}l|@{\quad}l}
        Name & Definition & Base \\
        \hline
        \rule{0mm}{2.7ex}$\CloneBF$ & All Boolean functions & $\{\land, \neg\}$ \\
        $\CloneR_0$ & $\{f : f \text{ is $0$-reproducing}\}$ & $\{\land, \not\limplies \}$ \\
        $\CloneR_1$ & $\{f : f \text{ is $1$-reproducing}\}$ & $\{\lor, \limplies \}$ \\
        $\CloneM$ & $\{f : f \text{ is monotone}\}$ & $\{\lor, \land, \false, \true\}$ \\
        $\CloneS_0$ & $\{f : f \text{ is $0$-separating}\}$ & $\{\limplies\}$ \\
        $\CloneS_1$ & $\{f : f \text{ is $1$-separating}\}$ & $\{\not\limplies \}$ \\
        $\CloneS_{00}$ & $\CloneS_0 \cap \CloneR_0 \cap \CloneR_1 \cap \CloneM$ & $ \{x\! \lor (y \!\land\! z)\} $ \\
        $\CloneS_{10}$ & $\CloneS_1 \cap \CloneR_0 \cap \CloneR_1 \cap \CloneM$ & $ \{x \!\land (y \!\lor\! z)\} $ \\
        $\CloneS_{11}$ & $\CloneS_1 \cap \CloneM$ & $ \{x \!\land (y \!\lor\! z),\false\} $ \\
        $\CloneD$ & $\{f : f \text{ is self-dual}\}$ & $\{ (x\! \land\! \overline{y}) \lor (x\! \land\! \overline{z}) \lor (\overline{y}\! \land\! \overline{z}) \}$ \\
        $\CloneD_2$ & $\CloneD \cap \CloneM$ & $ \{(x\!\land\! y) \lor (y\!\land\! z) \lor (x\!\land\! z)\}$ \\
        $\CloneL$ & $\{f : f \text{ is linear}\}$ & $\{ \xor,\true\}$ \\
        $\CloneL_0$ & $\CloneL \cap \CloneR_0$ & $\{ \xor\}$ \\
        $\CloneL_1$ & $\CloneL \cap \CloneR_1$ & $\{ \equiv\}$ \\
        $\CloneL_2$ & $\CloneL \cap \CloneR_0 \cap \CloneR_1$ & $\{ x \!\xor \!y \!\xor \!z\}$ \\
        $\CloneL_3$ & $\CloneL \cap \CloneD$ & $\{ x \!\xor\! y \!\xor\! z,  \neg \}$ \\
        $\CloneV$ & $\{f : f \equiv c_0 \lor \bigvee_{i=1}^n c_ix_i \text{ where the $c_i$s are constant}\}$ & $\{ \lor, \false,\true \}$ \\
        $\CloneV_2$ & $[\{\lor\}]$ & $\{ \lor \}$ \\
        $\CloneE$ & $\{f : f \equiv c_0 \land \bigwedge_{i=1}^n c_ix_i \text{ where the $c_i$s are constant}\}$ & $\{ \land, \false, \true \}$ \\
        $\CloneE_2$ & $[\{\land\}]$ & $\{ \land \}$ \\
        $\CloneN$ & $\{f : f \text{ depends on at most one variable}\}$ & $\{ \neg,\false,\true\}$ \\
        $\CloneN_2$ & $[\{\neg\}]$ & $\{ \neg\}$ \\
        $\CloneI$ & $\{f : f \text{ is a projection or a constant}\}$ & $\{\id, \false,\true\}$ \\
        $\CloneI_2$ & $[\{id\}]$ & $\{\id \}$ \\
      \end{tabular}
      
      \medskip
      \caption{
        \label{tab:clones}
        A list of Boolean clones with definitions and bases.
      }
    \end{table}

  For a finite set $B$ of Boolean functions,
  we define the \emph{Implication Problem} for $B$-formulae $\IMP(B)$
  as the following computational task: Given a set $A$ of $B$-formulae
  and a $B$-formula $\varphi$, decide whether $A\models \varphi$ holds.
  The complexity of the implication problem is classified in \cite{bemethvo08imp}.
  The results relevant to this paper are summarized in the following theorem.
  
  \begin{theorem}[{\cite[Theorem 4.1]{bemethvo08imp}}] \label{thm:implication}
    Let $B$ be a finite set of Boolean functions. Then $\IMP(B)$ is 
    \begin{enumerate}
      \item 
      $\co\NP$-complete if $\CloneS_{00} \subseteq [B]$, $\CloneS_{10} \subseteq [B]$ or $\CloneD_2 \subseteq [B]$, and  
      \item 
      in $\P$ for all other cases.
    \end{enumerate}
  \end{theorem}

\section{Default Logic} \label{sect:default_logic}
  
  Fix some finite set $B$ of Boolean functions and let $\alpha, \beta,\gamma$ be propositional $B$-formulae. A \emph{$B$-default (rule)} is an expression $d=\frac{\alpha:\beta}{\gamma}$; $\alpha$ is called \emph{prerequisite}, $\beta$ is called \emph{justification} and $\gamma$ is called \emph{consequent} of $d$. A \emph{$B$-default theory} is a pair $\encoding{W,D}$, where $W$ is a set of propositional $B$-formulae and $D$ is a finite set of $B$-default rules.
  Henceforth we will omit the prefix ``$B$-'' if $B= \CloneBF$ or the meaning is clear from the context.
  
  For a given default theory $\encoding{W,D}$ and a set of formulae $E$, let $\Gamma(E)$ be the smallest set of formulae such that 
  \begin{enumerate}
    \item \label{def_gamma_1} 
      $W\subseteq \Gamma(E)$,
    \item \label{def_gamma_2} 
      $\Gamma(E)$ is closed under deduction, \emph{i.e.}, $\Gamma(E)=\theorems{\Gamma(E)}$, and
    \item \label{def_gamma_3}
      for all defaults $\frac{\alpha:\beta}{\gamma} \in D$ with
      $\alpha \in \Gamma(E)$ and $\neg \beta \notin E$, it holds that
      $\gamma \in \Gamma(E)$. 
  \end{enumerate}
  A \emph{(stable) extension} of $\encoding{W,D}$ is a fix-point of $\Gamma$, \emph{i.e.}, a set $E$ such that $E=\Gamma(E)$.

  The following theorem by Reiter provides an alternative characterization of extensions:
  
  \begin{theorem}[{\cite{reiter:80}}] \label{thm:reiter-extensions}
    Let $\encoding{W,D}$ be a default theory and $E$ be a set of formulae. 
    \begin{enumerate}
      \item \label{thm:reiter-extensions-iterative}
      Let $E_0=W$ and
      $
        E_{i+1}= \theorems{E_i} \cup \big\{\gamma \,\big\arrowvert\, \frac{\alpha : \beta}{\gamma} \in D, \alpha \in E_i \text{ and } \neg \beta \notin E \big\}.
      $
      Then $E$ is a stable extension of $\encoding{W,D}$ if and only if $E= \bigcup_{i \in \N} E_i$.
      
      \item \label{thm:reiter-extensions-generating-defaults}
      Let $G= \big\{ \frac{\alpha:\beta}{\gamma} \in D \,\big\arrowvert\, \alpha \in E \text{ and } \neg \beta \notin E  \big\}$. If $E$ is a stable extension of $\encoding{W,D}$, then 
      \[
        \textstyle E = \theorems{W \cup \{\gamma \mid \frac{\alpha : \beta}{\gamma}  \in G\}}.
      \] 
      In this case, $G$ is also called the set of \emph{generating defaults} for $E$.
    \end{enumerate}
  \end{theorem}
  
  Observe that, as an immediate consequence of Theorem~\ref{thm:reiter-extensions}, stable extensions possess polynomial-sized witnesses, 
  namely the set of their generating defaults. Moreover, note that stable extensions need not be consistent. However, the
  following proposition shows that this only occurs if the set $W$ is
  already inconsistent. 
  \begin{proposition}[{\cite[Corollary 3.60]{matr93}}]
    Let $\encoding{W,D}$ be a default theory. Then $\allFormulae$ is a stable extension of $\encoding{W,D}$ if and only if $W$ is inconsistent.
  \end{proposition}

  As a consequence we obtain:
  \begin{corollary} \label{corr_cons_ext}
    Let $\encoding{W,D}$ be a default theory. 
    \begin{itemize}
      \item If $W$ is consistent, then every stable extension of $\encoding{W,D}$ is consistent.
      \item If $W$ is inconsistent, then $\encoding{W,D}$ has a stable
        extension.
    \end{itemize}
  \end{corollary}

  The main reasoning tasks in nonmonotonic logics give rise to the
  following three decision problems: 
  \begin{enumerate}
    \item
     the \emph{Extension Existence Problem}
    \dproblemdef{$\EXT(B)$}
        {a $B$-default theory $\encoding{W,D}$}
        {Does $\encoding{W,D}$ have a stable extension?}
    \item 
    the \emph{Credulous Reasoning Problem}
    \dproblemdef{$\CRED(B)$}
        {a $B$-formula $\varphi$ and a $B$-default theory $\encoding{W,D}$}
        {Is there a stable extension of $\encoding{W,D}$ that includes $\varphi$?}
    \item 
    the \emph{Skeptical Reasoning Problem}
    \dproblemdef{$\SKEP(B)$}
        {a $B$-formula $\varphi$ and a $B$-default theory $\encoding{W,D}$}
        {Does every stable extension of $\encoding{W,D}$ include $\varphi$?}
  \end{enumerate}
  
  The next theorem follows from \cite{gottlob:92} and states the complexity of the above decision problems for the general case $[B]=\CloneBF$.
  
  \begin{theorem} \label{thm:gottlob}
    Let $B$ be a finite set of Boolean functions such that $[B]=\CloneBF$. Then $\EXT(B)$ and $\CRED(B)$ are $\SigmaPtwo$-complete, whereas $\SKEP(B)$ is $\PiPtwo$-complete.
  \end{theorem}

  \begin{proof}
    The upper bounds given in \cite{gottlob:92} do not depend on the
    Boolean connectives allowed and thus hold for any finite set $B$
    of Boolean functions. For $\SigmaPtwo$- and $\PiPtwo$-hardness, it suffices to 
    note that if $[B] = \CloneBF$, then there exist $B$-formulae $f(x,y)$, $g(x,y)$ and $h(x)$
    such that $f(x,y)\equiv x \land y$, $g(x,y)\equiv x \lor y$, $h(x) \equiv \neg x$ 
    and both $x$ and $y$ occur at most once in $f$, $g$, and $h$
    \cite{lew79}. Hence, the hardness results generalize to arbitrary bases $B$ with $[B]=\CloneBF$. 
  \end{proof}

\section{The Complexity of Default Reasoning} \label{sect:default_reasoning}

In this section we will classify the complexity of the three problems
$\EXT(B)$, $\CRED(B)$, and $\SKEP(B)$ for all choices of Boolean
connectives $B$. We start with some preparations which will
substantially reduce the number of cases we have to consider.
  \begin{lemma} \label{lem:cloneB-equivcd-cloneB+1}
    Let $\Prob$ be any of the problems $\EXT$, $\CRED$, or $\SKEP$. Then
    for each finite set $B$ of Boolean functions,
    $\Prob(B) \equivcd \Prob(B\cup \{\true\})$.
  \end{lemma}
  
  \begin{proof}
    The reductions $\Prob(B) \leqcd \Prob(B \cup \{\true\})$ are
    obvious. 
    For the converse reductions, we will
    essentially substitute the constant $\true$
    by a new variable $t$ that is forced to be true (this trick goes
    back to Lewis \cite{lew79}). 
    For $\EXT$, the reduction is given by 
    $\encoding{W,D} \mapsto \encoding{W',D'}$, where 
    $W'=W_{[\true/t]} \cup \{t\}$, $D'=D_{[\true / t]}$, and $t$ is a
    new variable not occurring in $\encoding{W,D}$. 
    If $\encoding{W',D'}$ possesses a stable extension
    $E'$, then $t\in E'$. Hence, $E'_{[t/\true]}$ is
    a stable extension of $\encoding{W,D}$. On the other hand, if
    $E$ is a stable extension of $\encoding{W,D}$, then
    $\theorems{E_{[\true/ t]} \cup \{t\}} = E_{[\true/t]}$ is a stable 
    extension of $\encoding{W',D'}$.
    Therefore, each extension $E$ of $\encoding{W,D}$ corresponds to the
    extension $E_{[1/t]}$  of $\encoding{W',D'}$, and vice versa. 
  
    For the problems $\CRED$ and $\SKEP$, it suffices to note that the
    above reduction $\encoding{W,D} \mapsto \encoding{W',D'}$ has the
    additional property that for each formula $\varphi$ and each
    extension $E$ of $\encoding{W,D}$,
    $\varphi \in E$ if and only if $\varphi_{[1/ t]} \in E_{[1/t]}$.
  \end{proof}
  
  The next lemma shows that, quite often, $B$-default
  theories have unique extensions.
  \begin{lemma} \label{lem:extension-R1-M-unique}
    Let $B$ be a finite set of Boolean functions. 
    Let $\encoding{W,D}$ be a $B$-default theory. 
    If $[B] \subseteq \CloneR_1$ then $\encoding{W,D}$ has a unique stable extension.
    If $[B]\subseteq \CloneM$ then $\encoding{W,D}$ has at most one stable extension.
  \end{lemma}
  
  \begin{proof}
    For $[B] \subseteq \CloneR_1$, every premise, justification and
    consequent is $\true$-reproduc{\-}ing. As all consequences of 
    $\true$-reproducing functions are again 
    $\true$-reproducing and 
    the negation of a
    $\true$-reproducing function is not $\true$-reproducing,
    the justifications in $D$ become irrelevant. Hence
    the characterization of stable extensions from the first item in 
    Theorem~\ref{thm:reiter-extensions} simplifies to the
    following iterative construction:
    $E_0=W$ and
      $
        E_{i+1}= \theorems{E_i} \cup \big\{\gamma \,\big\arrowvert\,
        \frac{\alpha : \beta}{\gamma} \in D, \alpha \in E_i \big\}.
      $
    As $D$ is finite, this construction terminates after finitely many
    steps, \emph{i.e.}, $E_k=E_{k+1}$ for some $k\geq 0$.
    Then $E= \bigcup_{i \leq k} E_i$ is the unique stable extension of
    $\encoding{W,D}$. For a similar result confer \cite[Theorem 4.6]{bonoli02}.

    For  $[B] \subseteq \CloneM$, every formula is either
    $\true$-reproducing or equivalent to $\false$. As rules with
    justification equivalent to $\false$ are never applicable, each
    $B$-default theory $\encoding{W,D}$ with finite $D$ has at most one
    stable extension by the same argument as above.
  \end{proof}

  As an immediate corollary, the credulous and the skeptical reasoning
  problem are equivalent for the above choices of the underlying connectives.
  \begin{corollary} \label{cor:skep(b)-equiv-cred(b)}
    Let $B$ be a finite set of Boolean functions such that $[B] \subseteq \CloneR_1$ or $[B] \subseteq \CloneM$. Then $\CRED(B) \equivcd \SKEP(B)$.
  \end{corollary}

\subsection{The Extension Existence Problem}

  Now we are ready to classify the complexity of $\EXT$. The next
  theorem shows that this is a hexachotomy: the
  $\SigmaPtwo$-completeness of the general case \cite{gottlob:92} is
  inherited by all clones above $\CloneS_1$ and $\CloneD$;
  for monotone sets of connectives the complexity drops to $\DeltaPtwo$-completeness
  if $\land$, $\lor$ and $\false$ are available, and membership in $\P$ otherwise 
  (with this case splitting up into $\P$-completeness, $\NL$-completeness and triviality);
  lastly, for affine sets of connectives containing $\neg$ or $\false$ the complexity of $\EXT$ reduces to
  $\NP$-completeness.
  \begin{theorem} \label{thm:extension-existence}
    Let $B$ be a finite set of Boolean functions. Then $\EXT(B)$ is 
    \begin{enumerate}
      \item 
      \label{thm:extension-existence-SigmaP2}
      $\SigmaPtwo$-complete if $\CloneS_{1} \subseteq [B] \subseteq \CloneBF$ or $\CloneD \subseteq [B] \subseteq \CloneBF$,
      \item 
      \label{thm:extension-existence-DeltaP2}
      $\DeltaPtwo$-complete if $\CloneS_{11} \subseteq [B] \subseteq \CloneM$,
      \item 
      \label{thm:extension-existence-NP-complete}
      $\NP$-complete if $[B] \in \{\CloneN,\CloneN_2,\CloneL,\CloneL_0,\CloneL_3\}$, 
      \item 
      \label{thm:extension-existence-P}
      $\P$-complete if $[B] \in \{\CloneE,\CloneE_0,\CloneV,\CloneV_0\}$,
      \item 
      \label{thm:extension-existence-NL}
      $\NL$-complete if $[B] \in \{\CloneI,\CloneI_0\}$,
      and
      \label{thm:extension-existence-trivial}
      \item trivial in all other cases (\emph{i.e.}, if $[B] \subseteq \CloneR_1$).
    \end{enumerate}
  \end{theorem}
  
  The proof of Theorem~\ref{thm:extension-existence} will be established from the lemmas in this subsection.
  
  \begin{lemma}\label{lem:extension-existence-DeltaP2}
    Let $B$ be a finite set of Boolean functions such that $\CloneS_{11} \subseteq [B] \subseteq \CloneM$.
    Then $\EXT(B)$ is $\DeltaPtwo$-complete.
  \end{lemma}
  \begin{proof}
    We start by showing $\EXT(B) \in \DeltaPtwo$.
    Let $B$ be a finite set of Boolean functions such that $[B] \subseteq \CloneM$ and $\encoding{W,D}$ be a $B$-default theory.
    As the negated justification $\neg \beta$ of every default rule 
    $\frac{\alpha:\beta}{\gamma}\in D$ is either 
    equivalent to the constant $\true$ or not $\true$-reproducing,
    it holds that in the former case $\neg \beta$ is contained in any stable extension, 
    whereas in the latter $\neg \beta$ cannot be contained 
    in a consistent stable extension of $\encoding{W,D}$.
    We can distinguish between those two cases in polynomial time. 
    Therefore, using the characterization of 
    Theorem~\ref{thm:reiter-extensions}\,\eqref{thm:reiter-extensions-iterative}, 
    we can iteratively compute the applicable defaults and 
    test whether the premise of any default with unsatisfiable conclusion can be derived.%
    \begin{algorithm}[tb]
      \caption{Determining the existence of a stable extension}      
      \label{alg:extension-existence-dl}
      \begin{algorithmic}[1]
        \REQUIRE $\encoding{W,D}$
        \STATE $G_\mathrm{new} \leftarrow W$
        \REPEAT
          \STATE $G_\mathrm{old} \leftarrow G_\mathrm{new}$
          \FORALL{$\frac{\alpha:\beta}{\gamma} \in D$}
            \IF{$G_\mathrm{old} \models \alpha$ \textbf{and} $\beta \not\equiv \false$}
              \IF{$\gamma \equiv \false$}
                \RETURN \FALSE
              \ENDIF
              \STATE $G_\mathrm{new} \leftarrow G_\mathrm{new} \cup \{\gamma\}$
            \ENDIF
          \ENDFOR
        \UNTIL{$G_\mathrm{new} = G_\mathrm{old}$}
        \RETURN \TRUE
      \end{algorithmic}
    \end{algorithm}
    Algorithm~\ref{alg:extension-existence-dl} implements these
    steps on a deterministic Turing machine using a $\co\NP$-oracle to
    test for implication of $B$-formulae. Clearly, Algorithm~\ref{alg:extension-existence-dl}
    terminates after a polynomial number of steps. Hence, $\EXT(B)$ is contained in $\DeltaPtwo$.
    
    To show the $\DeltaPtwo$-hardness of $\EXT(B)$,
    we reduce from the $\DeltaPtwo$-complete problem $\SNSAT$~\cite[Theorem~3.4]{gottlob:95b} defined as follows:
    \problemdef{$\SNSAT$}
    {A sequence $(\varphi^i)_{1 \leq i \leq n}$ of formulae such that $\varphi^i$ contains the propositions $x_1,\ldots,x_{i-1}$ and $z_{i1},\ldots,z_{im_i}$}
    {Is $c_n = \true$, where $c_i$ is recursively defined via $c_i:=\true$ if and only if $\varphi^i$ is satisfiable by an assignment $\sigma$ such that $\sigma(x_j)=c_j$ for all $1 \leq j < i$?}
    
    Let $(\varphi^i)_{1 \leq i \leq n}$
    be the given sequence of propositional formulae and assume without loss of generality that $\varphi^i$ is in conjunctive normal form for all $1 \leq i \leq n$.
    For every proposition $x_j$ or $z_{ij}$ occurring in $(\varphi^i)_{1 \leq i \leq n}$, let 
    $x'_j$ respectively $z'_{ij}$ be a fresh proposition, 
    and define 
    \[
      \psi^i:=\varphi^i_{[
      \neg x_1/x'_1,\ldots,\neg x_{i-1}/x'_{i-1},
      \neg z_{i1}/z'_{i1},\ldots,\neg z_{im_i}/z'_{im_i}
      ]} 
      \land \bigwedge_{j=1}^{i-1} (x_j \lor x'_j) 
      \land \bigwedge_{j=1}^{m_i} (z_{ij} \lor z'_{ij}).
    \]
    The key observation in the relationship of $\varphi^i$ and $\psi^i$ is that, 
    for all $c_1,\ldots,c_{i-1} \in \{\false,\true\}$,
    $\varphi^i_{[x_1/c_1,\ldots,x_{i-1}/c_{i-1}]}$ is unsatisfiable if and only if 
    for each model $\sigma$ of $\psi^i_{[x_1/c_1,\ldots,x_{i-1}/c_{i-1},x'_1/\neg c_1,\ldots,x'_{i-1}/\neg c_{i-1}]}$ 
    there exists an index $1 \leq j \leq m_i$ 
    such that $\sigma$ sets to $\true$ both $z_{ij}$ and $z'_{ij}$. 
    We will use this observation to show that the $B$-default theory $\encoding{W,D}$ 
    defined below has a stable extension if and only if $(\varphi^i)_{1 \leq i \leq n}$
    is an instance of $\SNSAT$, that is, $\varphi^n_{[x_1/c_1,\ldots,x_{i-1}/c_{i-1}]}$ 
    is satisfiable for $c_1,\ldots,c_{i-1}$ recursively defined via
    \begin{equation} \label{eq:extension-existence-DeltaP2-1}
      c_i := \true \iff \varphi^i_{[x_1/c_1,\ldots,x_{i-1}/c_{i-1}]} \text{ is satisfiable}.
    \end{equation}
    Define $W:= \{\psi^1,\ldots,\psi^n\}$ and 
    \begin{align*}
      D&:= 
      \left\{
        \frac{\bigvee_{j=1}^{m_i} (z_{ij} \land z'_{ij}) \lor \bigvee_{j=1}^{i-1} (x_j \land x'_j) : \true}
        {x_i'} \,\middle|\, 1 \leq i < n
      \right\}
      \; \cup \\
      &\hphantom{\,:=\;} \left\{
        \frac{\bigvee_{j=1}^{m_n} (z_{nj} \land z'_{nj}) \lor \bigvee_{j=1}^{n-1} (x_j \land x'_j) : \true}
             {\false}
      \right\}.
    \end{align*}
    We will prove the claim appealing to the characterization of stable 
    extensions from Theorem~\ref{thm:reiter-extensions}\,\eqref{thm:reiter-extensions-iterative}.
    Let $E_0:=W$.
    If $\varphi^1$ is unsatisfiable then 
    $\frac{\bigvee_{j=1}^{m_1} (z_{1j} \land z'_{1j}) : \true}
          {x_1'}$ is applicable 
    and thus $x_1'$ is added to $E_1$.
    On the other hand, if $\varphi^1$ is satisfiable then there exists a model
    $\sigma$ of $\varphi^1$. Define $\hat\sigma$ as the extension of $\sigma$ 
    defined as $\hat\sigma(z'_{1j})=\neg\sigma(z_{1j})$ for all $1 \leq j \leq m_1$.
    By virtue of $\sigma \models \varphi^1$ and the construction of $\hat\sigma$, 
    we obtain that $\hat\sigma \models \psi^1$ while $\hat\sigma \not\models \bigvee_{j=1}^{m_1} (z_{1j} \land z'_{1j})$. Summarizing, $\varphi^1$ is unsatisfiable
    if and only if $\frac{\bigvee_{j=1}^{m_1} (z_{1j} \land z'_{1j}) : \true}{x_1'}$ is applicable. 

    Now suppose that $E_i$ is such that for all $j<i$ the proposition
    $x_j'$ is contained in $E_i$ if and only if $\varphi^j_{[x_1/c_1,\ldots,x_{j-1}/c_{j-1}]}$ with $c_1,\ldots,c_{j-1}$ defined as in \eqref{eq:extension-existence-DeltaP2-1} is unsatisfiable.
    If $\varphi^i_{[x_1/c_1,\ldots,x_{i-1}/c_{i-1}]}$ is unsatisfiable then 
    any model of the formula 
    \begin{equation}\label{eq:extension-existence-DeltaP2-2}
      \psi^i \land 
      \bigwedge_{\substack{1 \leq j < i,\\\sigma(c_j)=\true}} x_j \land 
      \bigwedge_{\substack{1 \leq j < i,\\\sigma(c_j)=\false}} x_j'
    \end{equation}
    sets to $\true$ both $z_{ij}$ and $z'_{ij}$ for some $1 \leq j \leq m_i$. 
    From \eqref{eq:extension-existence-DeltaP2-2} and the monotonicity of $\psi^i$, 
    we obtain that for each model $\sigma'$ of 
    $\psi^i \land \bigwedge_{1 \leq j < i,\sigma(c_j)=\false} x_j'$
    there must exist either an index $1 \leq j < i$ such that $\sigma'$
    sets $x_j$ and $x_j'$ to $\true$, or an index $1 \leq j \leq m_i$ such that $\sigma'$ sets $z_{ij}$ and $z'_{ij}$ to $\true$.
    Consequently, $\frac{\bigvee_{j=1}^{m_i} (z_{ij} \land z'_{ij}) \lor \bigvee_{j=1}^{i-1} (x_j \land x'_j) : \true}{x_i'}$ is applicable and $x_i' \in E_{i+1}$.
    On the other hand, if $\varphi^i_{[x_1/c_1,\ldots,x_{i-1}/c_{i-1}]}$ is satisfiable
    then there exists a model $\sigma$ that can be extended to $\hat\sigma$ by
    $\hat\sigma(z'_{ij})=\neg\sigma(z_{ij})$ for all $1 \leq j \leq m_i$ and
    $\hat\sigma(x'_{j})=\neg\sigma(x_{j})$ for all $1 \leq j < i$ 
    such that $\hat\sigma \models \psi^i$ and $\hat\sigma \not\models \bigvee_{j=1}^{m_i} (z_{ij} \land z'_{ij}) \lor \bigvee_{j=1}^{i-1} (x_j \land x'_j)$.  
    Summarizing, $\varphi^i$ is unsatisfiable if and only if $\frac{\bigvee_{j=1}^{m_i} (z_{ij} \land z'_{ij}) \lor \bigvee_{j=1}^{i-1} (x_j \land x'_j) : \true}{x_i'}$ is applicable.
    
    The direction from right to left now follows from 
    the fact that $\varphi_n$ is satisfiable if and only if 
    $\frac{\bigvee_{j=1}^{m_n} (z_{ij} \land z'_{ij}) \lor \bigvee_{j=1}^{n-1} (x_i \land x'_i) : \true}{\false}$ is not applicable, 
    which in turn implies that $\encoding{W,D}$ has a stable extension.
    Conversely, if $\varphi^n_{[x_1/c_1,\ldots,x_{n-1}/c_{n-1}]}$ is unsatisfiable with $c_1,\ldots,c_{n-1}$ defined as in \eqref{eq:extension-existence-DeltaP2-1}, then any model of 
    $\psi^i \land \bigwedge_{1 \leq j < i,\sigma(c_j)=\false} x_j'$
    sets to true either $x_j$ and $x_j'$ for some $1 \leq j < i$ or 
    $z_{ij}$ and $z'_{ij}$ for some $1 \leq j \leq m_i$.
    As a result, the default $\frac{\bigvee_{j=1}^{m_n} (z_{ij} \land z'_{ij}) \lor \bigvee_{j=1}^{n-1} (x_j \land x'_j) : \true}{\false}$ is applicable and $\encoding{W,D}$ does not possess a stable extension.
    
    Finally, observe that all formulae contained in $\encoding{W,D}$ are monotone. 
    Hence, $\encoding{W,D}$ is a $\{\land,\lor,\false,\true\}$-default theory. 
    Let $B$ be a finite set of Boolean functions such that $\CloneS_{11} \subseteq [B]$.
    Replacing all occurrences of $x \land y$ and $x \lor y$ in $\encoding{W,D}$ 
    with their respective $(B \cup \{\true\})$-representations $f_\land(x,y)$ and $f_\lor(x,y)$,
    and eliminating the constant $\true$ as in the proof of Lemma~\ref{lem:cloneB-equivcd-cloneB+1}
    yields a $B$-default theory $\encoding{W^B,D^B}$ that is equivalent to $\encoding{W,D}$. 
    The variables $x$ or $y$ may occur several times in the body of $f_\land$ or $f_\lor$, 
    hence $\encoding{W^B,D^B}$ might be exponential in the length of the original input. 
    To avoid this blowup, we exploit the associativity of $\land$ and $\lor$:
    we insert parentheses such that the conjunctions and disjunctions in each of the above 
    formulae are transformed into trees of logarithmic depth. 
    
    Thus we have established a reduction from $\SNSAT$ to $\EXT(B)$ for all $B$ such that $\CloneS_{11} \subseteq [B]$.
    This concludes the proof.
  \end{proof}

  \begin{lemma} \label{lem:extension-existence-NP}
    Let $B$ be a finite set of Boolean functions such that $[B] \in \{\CloneN, \CloneN_2,$ $\CloneL, \CloneL_0, \CloneL_3\}$.
    Then $\EXT(B)$ is $\NP$-complete.
  \end{lemma}
  \begin{proof}
    We start by showing $\EXT(B) \in \NP$ for $[B] \subseteq \CloneL$. Given a
    default theory $\encoding{W,D}$, we first guess a set $G\subseteq
    D$ which will serve as the set of generating defaults for a stable
    extension.  
    Let $G'=W \cup \{\gamma \mid \frac{\alpha : \beta}{\gamma}  \in G\}$.
    We use Theorem~\ref{thm:reiter-extensions} to verify whether
    $\theorems{G'}$ is indeed a stable extension of $\encoding{W,D}$.
    For this we inductively compute generators $G_i$ for the sets
    $E_i$ from Theorem~\ref{thm:reiter-extensions},
    until eventually $E_i=E_{i+1}$ (note, that because
    $D$ is finite, this always occurs).
    We start by setting $G_0=W$. Given $G_i$, we check for each rule
    $\frac{\alpha:\beta}{\gamma} \in D$, whether $G_i \models \alpha$ 
    and $G' \not\models \neg\beta$ (as all formulae belong to
    $\allFormulae(B)$, this is possible by
    Theorem~\ref{thm:implication}). If so, then $\gamma$ is put into 
    $G_{i+1}$. If this process terminates, \emph{i.e.}, if $G_i=G_{i+1}$,
    then we check whether $G'=G_i$. 
    By Theorem~\ref{thm:reiter-extensions}, this test is positive if and only if
    $G$ generates a stable extension of $\encoding{W,D}$.

    To show $\NP$-hardness of $\EXT(B)$ for $\CloneN \subseteq [B]$, we will
    $\leqcd$-reduce $\ThreeSAT$ to $\EXT(B)$. 
    Let $\varphi = \bigwedge_{i=1}^n (\ell_{i1} \lor \ell_{i2} \lor \ell_{i3})$
    and $\ell_{ij}$ be literals over propositions $\{x_1,\ldots,x_m\}$ for 
    $1 \leq i \leq n$, $1 \leq j \leq 3$. We transform $\varphi$ to the 
    $B$-default theory $\encoding{W,D_\varphi}$, where $W:=\emptyset$ and 
    \begin{align*}
      D_\varphi := & \bigg\{ \frac{\true : x_i}{x_i} \,\bigg\arrowvert\, 1 \leq i \leq m \bigg\} \cup 
            \bigg\{ \frac{\true : \neg x_i}{\neg x_i} \,\Big\arrowvert\, 1 \leq i \leq m \bigg\} \cup \\
          & \bigg\{ \frac{\sneg{\ell}_{i\pi(1)} : \sneg{\ell}_{i\pi(2)}}{\ell_{i\pi(3)}} \,\bigg\arrowvert\, 1 \leq i
            \leq n, \pi \text{ is a permutation of } \{1,2,3\} \bigg\} \enspace.
    \end{align*} 
    To prove the correctness of the reduction, 
    first assume $\varphi$ to be satisfiable. For each 
    satisfying assignment $\sigma:\{x_1,\dots,x_m\} \to \{0,1\}$ for $\varphi$, we claim that
    \[ 
      E:= \theorems{\{ x_i \mid \sigma(x_i)=\true \}
        \cup \{ \neg x_i \mid \sigma(x_i)=\false \}}
    \]
    is a stable extension of $\encoding{W,D_\varphi}$.
    We will verify this claim with the help of the first part of
    Theorem~\ref{thm:reiter-extensions}. 
    Starting with $E_0=\emptyset$, we already get $E_1=E$
    by the default rules 
    $\frac{\true : x_i}{x_i}$ and 
    $\frac{\true : \neg x_i}{\neg x_i}$ in $D_\varphi$. 
    As $\sigma$ is a satisfying assignment for
    $\varphi$, each consequent of a default rule in $D_\varphi$ is already
    in $E$. Hence $E_2=E_1$ and therefore $E=\bigcup_{i\in\N} E_i$ is 
    a stable extension of $\encoding{W,D_\varphi}$.

    Conversely, assume that $E$ is a stable extension of
    $\encoding{W,D_\varphi}$. 
    Because of the default rules $\frac{\true : x_i}{x_i}$ and 
    $\frac{\true : \neg x_i}{\neg x_i}$, we either get
    $x_i\in E$ or $\neg x_i\in E$ for all $i=1,\dots, m$.
    The rules of the type 
    $\frac{\sneg{\ell}_{i1} : \sneg{\ell}_{i2}}{\ell_{i3}}$ ensure
    that $E$ contains at least one literal from each clause 
    $\ell_{i1}\vee \ell_{i2} \vee \ell_{i3}$ in $\varphi$.
    As $E$ is deductively closed, $E$ contains $\varphi$. 
    By Corollary~\ref{corr_cons_ext}, the extension $E$ is consistent, and therefore
    $\varphi$ is satisfiable. 
    
    Hence, $\EXT(B)$ is $\NP$-complete for every finite set $B$ such that $\CloneN \subseteq [B] \subseteq \CloneL$.
    The remaining cases $[B]\in \{\CloneN_2,\CloneL_0,\CloneL_3\}$ follow from Lemma~\ref{lem:cloneB-equivcd-cloneB+1}, 
    because     
    $[\CloneN_2 \cup \{\true\}]=\CloneN$,
    $[\CloneL_0 \cup \{\true\}]=\CloneL$, and 
    $[\CloneL_3 \cup \{\true\}]=\CloneL$.     
  \end{proof}

  \begin{lemma}\label{lem:extension-existence-P}
    Let $B$ be a finite set of Boolean functions such that 
    $[B] \in \{\CloneE,\CloneE_0,\linebreak[1]\CloneV,\CloneV_0\}$. 
    Then $\EXT(B)$ is $\P$-complete.
  \end{lemma}
  \begin{proof}
    Let $B$ be a finite set of Boolean functions such that 
    $[B] \in \{\CloneE,\CloneE_0,\CloneV,\CloneV_0\}$. 
    Membership in $\P$ is is obtained from Algorithm~\ref{alg:extension-existence-dl}, 
    as for these types of $B$-formulae, we have an efficient 
    test for implication.
    
    To prove $\P$-hardness for $\CloneE_0 \subseteq [B]$,
    we provide a reduction from the complement of the accessibility problem for directed
    hypergraphs, $\overline{\HGAP}$. 
    In directed hypergraphs $H=(V,F)$, hyperedges $e \in F$ consist of a
    set of source nodes $\mathrm{src}(e) \subseteq V$ and a 
    destination $\mathrm{dest}(e) \in V$.
    Instances of $\HGAP$ contain a directed hypergraph $H=(V,F)$, a
    set $S \subseteq V$ of source nodes, and a target node $t\in V$.
    $\HGAP$ is $\P$-complete under
    $\leqcd$-reductions \cite{sriy90}, even if 
    restricted to hypergraphs whose edges contain at most two source nodes.
    
    We transform a given instance $(H,S,t)$ to 
    the $\EXT(\{\land,\false,\true\})$-instance $\encoding{W,D}$ with
    \[
      W := \{p_s \mid s \in S\},\ \
      D := \bigg\{ \frac{\bigwedge_{v \in \src{e}} p_v  : \true }
                        {p_{\dest{e}}} 
                   \,\bigg|\, e \in F
            \bigg\} 
            \cup
            \bigg\{ \frac{p_t  : \true }
                         {\false} 
            \bigg\} 
    \]
    with pairwise distinct propositions $p_v$ for $v \in V$.
    It is easy to verify that $(H,s,t) \in \HGAP \iff \encoding{W,D} \notin \EXT(\{\land,\false,\true\})$. 
    Using Lemma~\ref{lem:cloneB-equivcd-cloneB+1} and replacing $\land$ by its $B$-representation, 
    we obtain $\overline{\HGAP} \leqcd \EXT(B)$ for all finite sets $B$ such that $\CloneE_0 \subseteq [B]$. As $\P$ is closed under complementation, $\EXT(B)$ is $\P$-complete.

    For $\CloneV_0 \subseteq [B]$, set
    \[
      W := \Big\{ \bigvee_{s \notin S} p_s \Big\},\
      D := \bigg\{ \frac{\bigvee_{v \in V \setminus \src{e}} p_v  : \true } 
                        {\bigvee_{v \in V \setminus (\src{e} \cup \{\dest{e}\})} p_v} 
                   \,\bigg|\, e \in F
           \bigg\} 
           \cup
           \bigg\{ \frac{\bigvee_{v \in V\setminus\{t\}}  p_v  : \true }
                        {\false}
           \bigg\}.
    \]
    We claim that this mapping realizes the reduction 
    $\overline{\HGAP} \leqcd \EXT(\{\lor,\false,\true\})$.
    First suppose that $t$ can be reached from $S$ in $H$. Then there exists
    a sequence $(S_i)_{0 \leq i \leq n}$ of sets of nodes  
    such that $S_0=S$, $t \in S_n$, and for all $0 \leq i <n$,
    $S_{i+1}$ is obtained from $S_i$ by  
    adding the destination $\mathrm{dest}(e)$ of a hyperedge $e \in F$
    with $\mathrm{src}(e) \subseteq S_i$. 
    Let $(e_i)_{0 \leq i < n}$ denote the corresponding sequence of hyperedges used to obtain $S_{i+1}$ from $S_i$.
    Then, for all $0 \leq i < n$, the following holds:
    \[
      \src{e_i} \subseteq S_i \iff \frac{\bigvee_{v \in V \setminus \src{e}} p_v  : \true } 
                        {\bigvee_{v \in V \setminus (\src{e} \cup \{\dest{e}\})} p_v} \text{ is applicable in } E_i,
    \]
    where $(E_{i})_{i \in \N}$ is the sequence from Theorem~\ref{thm:reiter-extensions}\,\eqref{thm:reiter-extensions-iterative}.
    As $\bigcup_{i\in \N} E_i$ is guaranteed to be unique by Lemma~\ref{lem:extension-R1-M-unique} and $t \in S_{n}$, we obtain that $\false \in E_{n+1}$. Consequently, $\encoding{W,D}$ does not possess a stable extension.

    Conversely, if $\encoding{W,D}$ does not admit a stable extension, then $\false$ has to be derivable. Accordingly, there exists a sequence of defaults $(d_i)_{0 \leq i \leq n}$ such that the premise of $d_i$ can be derived from 
    $W \cup \big\{\gamma \,\big|\, d_j=\frac{\alpha:\beta}{\gamma}, 0 \leq j <i\big\}$ and
    $d_n = \frac{\bigvee_{v \in V\setminus\{t\}}  p_v  : \true }{\false}$.
    By construction of $\encoding{W,D}$, this sequence can be
    translated into a sequence $(S_i)_{0 \leq i \leq n}$ of node sets in the hypergraph
    such that $S_0=S$, $t \in S_n$, and for all $0 \leq i <n$,
    $S_{i+1}$ is obtained from $S_i$ by  
    adding the destination $\dest{e}$ of a hyperedge $e \in F$
    with $\src{e} \subseteq S_i$. 
    Consequently, $t$ is reachable from $S$ in $H$ and we conclude
    that $\overline{\HGAP} \leqcd \EXT(\{\lor,\false,\true\})$. 
    Using Lemma~\ref{lem:cloneB-equivcd-cloneB+1}, we get
    $\overline{\HGAP} \leqcd \EXT(\{\lor,\false\})$.

    To see that $\EXT(\{\lor,\false\}) \leqcd \EXT(B)$ for all finite sets $B$ 
    such that $\CloneV_0 \subseteq [B]$, 
    we proceed as in the proof of Lemma~\ref{lem:extension-existence-DeltaP2} 
    and insert parentheses such that the disjunctions in each of the above 
    formulae are transformed into tree of logarithmic depth. Hence, 
    replacing all occurrences of $\lor$ in $W$, $D$ and $\varphi$ with its $B$-representation 
    yields an $\EXT(B)$-instance of size polynomial in the original input.    
    Concluding, $\EXT(B)$ is $\P$-complete.
  \end{proof}
  
  \begin{lemma}\label{lem:extension-existence-NL}
    Let $B$ be a finite set of Boolean functions such that 
    $[B] \in \{\CloneI,\CloneI_0\}$. 
    Then $\EXT(B)$ is $\NL$-complete with respect to constant-depth reductions.
  \end{lemma}
  \begin{proof}
    Let $B$ be a finite set of Boolean functions such that $[B] \in \{\CloneI,\CloneI_0\}$.
    We will first show membership in $\NL$ by giving a reduction to the complement of the graph accessibility problem, $\overline{\GAP}$.
    
    Let $\encoding{W,D}$ be a $B$-default theory.
    Analogously to the proof of Lemma~\ref{lem:extension-existence-DeltaP2}, 
    it holds that $\encoding{W,D}$ has a stable extension if and only if 
    either $W$ is inconsistent or the conclusions of all applicable defaults are consistent.
    Assume that $W$ is consistent
    and denote by $D' \subseteq D$ those defaults $\frac{\alpha:\beta}\gamma \in D$ with $\beta \not\equiv \false$.
    Then a $B$-default rule $\frac{\alpha:\beta}\gamma \in D'$ is applicable if and only if the proposition 
    $\alpha$ is contained in $W$ or itself the conclusion of an applicable default. 
    Therefore, testing whether the conclusions of all applicable defaults are 
    consistent is essentially equivalent to solving a reachability problem in a directed graph.
    Define $G_{\encoding{W,D}}$ as the directed graph $(V,F)$ with 
    \begin{align*}
      V &:= \textstyle \{\false,\true\} \cup W \cup \left\{\alpha,\gamma \,\middle|\, \frac{\alpha:\beta}\gamma \in D \right\}, \\
      F &:= \textstyle \{ (\true,x) \,|\, x \in W \} \cup \left\{ (\alpha,\gamma) \middle| \frac{\alpha:\beta}\gamma \in D, \beta \not\equiv \false \right\} \\
    \shortintertext{if $W$ is consistent, and }
      V &:= \{\false,\true\}, \\
      F &:= \emptyset
    \end{align*}
    otherwise.
    It is easy to see that $\encoding{W,D}$ has a stable extension if and only if 
    there is no path from $\true$ to $\false$ in $G_{\encoding{W,D}}$.
    Thus the function mapping the given $B$-default theory $\encoding{W,D}$ to the $\GAP$-instance $(G_{\encoding{W,D}},\true,\false)$
    constitutes a reduction from $\EXT(B)$ to $\overline{\GAP}$. 
    As the consistency of $W$ can be determined in $\AC0$, the reduction can be computed using constant-depth circuits. 
    Membership in $\NL$ follows from the closure of $\NL$ under complementation.

    To show $\NL$-hardness, we establish a constant-depth reduction 
    in the converse direction.
    For a directed graph $G=(V,F)$ and two nodes $s,t\in V$, we
    transform the given $\GAP$-instance $(G,s,t)$ to $\encoding{W,D}$ with
    \[
      \textstyle
      W       := \{p_s\},\ \
      D       := \Big\{ \frac{p_u : p_u}{ p_v } \,\Big|\, (u,v) \in F \Big\} \cup 
                 \Big\{ \frac{p_t : p_t}{ \false } \Big\}
    \]
    Clearly, $(G,s,t) \in \GAP$ if and only if $\encoding{W,D}$ does not have a stable extension.
    As $\NL$ is closed under complementation, the lemma is established.
  \end{proof}

  \begin{proof}[Theorem~\ref{thm:extension-existence}]
    For $\CloneS_{1} \subseteq [B] \subseteq \CloneBF$ or $[B]=
    \CloneD$, observe that in both cases $\CloneBF = [B \cup \{\true\}]$.  
    Claim~\ref{thm:extension-existence-SigmaP2} then follows from
    Theorem~\ref{thm:gottlob} and Lemma~\ref{lem:cloneB-equivcd-cloneB+1}.
    Claims two to five are established in Lemmas~\ref{lem:extension-existence-DeltaP2}--\ref{lem:extension-existence-NL}.
    For all sets $B$ not captured by the above, it now holds that $[B] \subseteq \CloneR_1$.
    Thus, the sixth claim follow directly from Lemmas~\ref{lem:extension-R1-M-unique}, 
  \end{proof}

\subsection{The Credulous and the Skeptical Reasoning Problem}
  
  We will now analyse the credulous and the skeptical reasoning problems.
  For these problems, there are two sources of complexity. On the one
  hand, we need to determine a candidate for a stable extension. On
  the other hand, we have to verify that this candidate is indeed a
  finite characterization of some stable extension that includes a
  given formula\,---\,a task that requires to test for formula implication. 
  Depending on the Boolean connectives allowed, one or both tasks can be performed in polynomial time. 
  We obtain the full complexity, \emph{i.e.},  
  $\SigmaPtwo$-completeness for $\CRED(B)$ and $\PiPtwo$-completeness 
  for $\SKEP(B)$, where both problems $\EXT(B)$ and $\IMP(B)$ attain their highest
  complexity. 
  The complexity reduces to $\DeltaPtwo$ for
  clones that allow for an efficient computation of stable extensions
  but whose implication problem remains $\co\NP$-complete. More precisely, 
  the problem is $\DeltaPtwo$-complete if $\false \in [B]$
  and becomes $\co\NP$-complete otherwise.
  Conversely, if the implication problem becomes easy
  but determining an extension candidate is hard, then
  $\CRED(B)$ is $\NP$-complete, while $\SKEP(B)$ is $\co\NP$-complete.
  This is the case for $[B]\in\{\CloneN,\CloneN_2,\CloneL,\CloneL_0,\CloneL_3\}$.
  Finally, for clones $B$ that allow for solving both tasks in
  polynomial time, both $\CRED(B)$ and $\SKEP(B)$ are in $\P$. 

  The complete classification of $\CRED(B)$ is given in the following theorem.
    \begin{theorem} \label{thm:credulous-reasoning}
    Let $B$ be a finite set of Boolean functions. Then $\CRED(B)$ is 
    \begin{enumerate}
      \item \label{thm:credulous-reasoning-SigmaP2}
      $\SigmaPtwo$-complete if $\CloneS_{1} \subseteq [B] \subseteq \CloneBF$ or $\CloneD \subseteq [B] \subseteq \CloneBF$,
      \item \label{thm:credulous-reasoning-DeltaP2}
      $\DeltaPtwo$-complete if $\CloneS_{11} \subseteq [B] \subseteq \CloneM$, 
      \item \label{thm:credulous-reasoning-coNP}
      $\co\NP$-complete if $X \subseteq [B] \subseteq \CloneR_1$ for $X \in \{\CloneS_{00}, \CloneS_{10}, \CloneD_2\}$, 
      \item \label{thm:credulous-reasoning-NP}
      $\NP$-complete if $[B] \in \{\CloneN,\CloneN_2,\CloneL,\CloneL_0,\CloneL_3\}$,
      \item \label{thm:credulous-reasoning-P}
      $\P$-complete if $\CloneV_2 \subseteq [B] \subseteq \CloneV$, $\CloneE_2 \subseteq  [B] \subseteq \CloneE$ or $[B]\in\{\CloneL_1,\CloneL_2\}$,
      and
      \item \label{thm:credulous-reasoning-NL}
      $\NL$-complete if $\CloneI_2 \subseteq [B] \subseteq \CloneI$. 
    \end{enumerate}
  \end{theorem}
  
  The proof of Theorem~\ref{thm:credulous-reasoning} follows from the upper and lower bounds given in Propositions~\ref{prop:credulous-reasoning-upper-bounds} and \ref{prop:credulous-reasoning-lower-bounds} below.
  
  \begin{proposition} \label{prop:credulous-reasoning-upper-bounds}
    Let $B$ be a finite set of Boolean functions. Then $\CRED(B)$ is contained 
    \begin{enumerate}
      \item \label{upper-bounds-credulous-reasoning-SigmaP2}
      in $\SigmaPtwo$ if $\CloneS_{1} \subseteq [B] \subseteq
        \CloneBF$ or $\CloneD \subseteq [B] \subseteq \CloneBF$, 
      \item \label{upper-bounds-credulous-reasoning-DeltaP2}
      in $\DeltaPtwo$ if $[B] \subseteq \CloneM$, 
      \item \label{upper-bounds-credulous-reasoning-coNP}
      in $\co\NP$ if $[B] \subseteq \CloneR_1$, 
      \item \label{upper-bounds-credulous-reasoning-NP}
      in $\NP$ if $[B] \subseteq \CloneL$,
      \item \label{upper-bounds-credulous-reasoning-P}
      in $\P$ if $[B] \subseteq \CloneV$, 
      $[B] \subseteq \CloneE$ or $[B] \subseteq \CloneL_1$,
      and
      \item \label{upper-bounds-credulous-reasoning-NL}
      in $\NL$ if $[B] \subseteq \CloneI$. 
    \end{enumerate}
  \end{proposition}
  
  \begin{proof}
    Part one follows from
    Theorem~\ref{thm:gottlob} and Lemma~\ref{lem:cloneB-equivcd-cloneB+1}. 
    
    For $[B] \subseteq \CloneM$, 
    membership in $\DeltaPtwo$ is obtained from a straightforward 
    extension of Algorithm~\ref{alg:extension-existence-dl}: 
    first iteratively compute the applicable defaults $G$ 
    while asserting that $\encoding{W,D}$ has a stable extension using 
    Algorithm~\ref{alg:extension-existence-dl}, and eventually verify 
    that $\varphi$ is implied by $W$ and the conclusions in $G$.
    
    For $[B]\subseteq \CloneR_1$, the justifications $\beta$ are
    irrelevant for computing a stable extension, as for every default rule
    $\frac{\alpha:\beta}{\gamma}\in D$ we cannot derive 
    $\lnot\beta$ ($\lnot\beta$ is not $\true$-reproducing).
    Hence, a unique consistent stable extension $E$ is guaranteed to exist by 
    Theorem~\ref{lem:extension-R1-M-unique}.
    Using Algorithm~\ref{alg:extension-existence-dl}
    we can iteratively compute the generating defaults of $E$ of the unique 
    consistent stable extension of $\encoding{W,D}$ and eventually
    check whether $\varphi$ is implied by $W$ and the conclusions in of the generating defaults of $E$.
    
    For $[B] \subseteq \CloneL$, we proceed similarly as in the proof
    of part~\ref{thm:extension-existence-NP-complete} in Theorem~\ref{thm:extension-existence}. First, we guess 
    a set $G$ of generating defaults and subsequently verify that both
    $\theorems{W \cup \{\gamma \mid \frac{\alpha:\beta}{\gamma} \in G\}}$ is a
    stable extension and that $W \cup \{\gamma \mid
    \frac{\alpha:\beta}{\gamma} \in G\}\models\varphi$. Using
    Theorem~\ref{thm:implication}, both
    conditions may be verified in polynomial time.
    
    For $[B] \subseteq \CloneV$, $[B] \subseteq \CloneE$,  
    and $[B]\subseteq \CloneL_1$, we again use 
    Algorithm~\ref{alg:extension-existence-dl}. 
    As for these types of $B$-formulae
    we have an efficient test for implication 
    (Theorem~\ref{thm:implication}), we get $\CRED(B)\in \P$.

    For $[B]\subseteq \CloneI$, 
    observe that $\NL$ is closed under intersection.
    Hence, given a $B$-default theory $\encoding{W,D}$ and a $B$-formula $\varphi$ we can first 
    test whether $\encoding{W,D}$ has a stable extension $E$ using Lemma~\ref{lem:extension-existence-NL} and 
    subsequently assert that $\varphi \in E$ by reusing the graph $G_{\encoding{W,D}}$ constructed from $\encoding{W,D}$: 
    it holds that $\varphi \in E$ if and only if the node corresponding to $\varphi$ is contained in $G_{\encoding{W,D}}$ and reachable from the node $\true$. Thus, $\CRED(B) \in \NL$.
  \end{proof}

  We will now establish the lower bounds required to complete the proof of Theorem \ref{thm:credulous-reasoning}. 
  
  \begin{proposition} \label{prop:credulous-reasoning-lower-bounds}
    Let $B$ be a finite set of Boolean functions. 
    Then $\CRED(B)$ is 
    \begin{enumerate}
      \item \label{prop:credulous-reasoning-lower-bounds-sigmaP2}
      $\SigmaPtwo$-hard if $\CloneS_{1} \subseteq [B]$ or $\CloneD \subseteq [B]$,
      \item \label{lem:credulous-reasoning-lower-bounds-deltaP2}
      $\DeltaPtwo$-hard if $\CloneS_{11} \subseteq [B]$,
      \item \label{prop:credulous-reasoning-lower-bounds-coNP}
      $\co\NP$-hard if $\CloneS_{00} \subseteq [B]$, $\CloneS_{10} \subseteq [B]$ or $\CloneD_2 \subseteq [B]$, 
      \item \label{prop:credulous-reasoning-lower-bounds-NP}
      $\NP$-hard if $\CloneN_2 \subseteq [B]$ or $\CloneL_0 \subseteq [B]$,
      \item \label{prop:credulous-reasoning-lower-bounds-P}
      $\P$-hard if $\CloneV_2 \subseteq [B]$, 
      $\CloneE_2 \subseteq  [B]$ or $\CloneL_2 \subseteq [B]$,
      and
      \item \label{prop:credulous-reasoning-lower-bounds-NL}
      $\NL$-hard for all other clones.
    \end{enumerate}
  \end{proposition}

  \begin{proof}
    Part one follows from
    Theorem~\ref{thm:gottlob} and Lemma~\ref{lem:cloneB-equivcd-cloneB+1}.
    
    For the second part, observe that the constant $\true$ is contained in any stable extension. 
    The second part thus follows from Lemmas~\ref{lem:cloneB-equivcd-cloneB+1} and~\ref{lem:extension-existence-DeltaP2}.

    For $\CloneS_{00} \subseteq [B]$, $\CloneS_{10} \subseteq [B]$, and
    $\CloneD_2 \subseteq [B]$, 
    $\co\NP$-hardness is established by
    a $\leqcd$-reduction from $\IMP(B)$. 
    Let $A\subseteq \allFormulae(B)$ and
    $\varphi\in\allFormulae(B)$. 
    Then the default theory $\encoding{A,\emptyset}$ has the
    unique stable extension $\theorems{A}$, and hence 
    $A \models \varphi$ if and only if $(\encoding{A,\emptyset},\varphi)\in \CRED(B)$.
    Therefore, $\IMP(B) \leqcd \CRED(B)$, and the claim follows with
    Theorem~\ref{thm:implication}. 

    For the fourth part, it suffices to prove $\NP$-hardness for
    $\CloneN_2 \subseteq [B]$. 
    For $\CloneL_0 \subseteq [B]$, the claim then follows by
    Lemma~\ref{lem:cloneB-equivcd-cloneB+1}.
    For $\CloneN_2 \subseteq [B]$, we obtain $\NP$-hardness of $\CRED(B)$ 
    by adjusting the reduction given in the
    proof of item~\ref{thm:extension-existence-NP-complete} of
    Theorem~\ref{thm:extension-existence}.
    Consider the mapping $\varphi \mapsto (\encoding{\{\psi\},D_\varphi},\psi)$, where
    $D_\varphi$ is the set of default rules constructed from $\varphi$ in
    Theorem~\ref{thm:extension-existence}, and $\psi$ is a satisfiable
    $B$-formula such that $\varphi$ and
    $\psi$ do not use common variables. 
    By Theorem~\ref{thm:extension-existence}, $\varphi \in \ThreeSAT$
    if and only if $\encoding{\{\psi\},D_\varphi}$ has a stable extension. As any
    extension of $\encoding{\{\psi\},D_\varphi}$ contains $\psi$, 
    we obtain $\ThreeSAT \leqcd \CRED(B)$ via the above reduction.
    
    For the fifth part, the cases $\CloneE_2 \subseteq [B]$ and 
    $\CloneV_2 \subseteq [B]$ follow similarly from
    Lemmas~\ref{lem:cloneB-equivcd-cloneB+1} and~\ref{lem:extension-existence-P}.
    It hence suffices to prove the $\P$-hardness for $[B]\in\{\CloneL_1,\CloneL_2\}$
    We again provide a reduction from $\HGAP$ restricted to hypergraphs 
    whose edges contain at most two source nodes.
    To this end, we transform a given instance $(H,S,t)$ with $H=(V,F)$, to the 
    $\CRED(\{x \xor y \xor z,\true\})$-instance $(\encoding{W,D},\varphi)$, where
    \begin{align*}
      W &:= \{ p_s \mid s \in S \}, \\
      D &:= \bigg\{ \frac{p_{ \src{e}} : \true }
                         {p_{ \dest{e}}} 
                    \,\bigg\arrowvert\, e \in F, |\src{e}|=1 \bigg\}\,\cup \\
         &\hphantom{\;:=\;} 
            \bigg\{ \frac{p_{\mathrm{src}_1(e)}  :  \true }
                         {p_{e}},
                    \frac{p_{\mathrm{src}_2(e)}  : \true }
                         {p_{e}}, 
                    \frac{p_{\mathrm{src}_1(e)} \xor p_{\mathrm{src}_2(e)} \xor p_{e}  : \true }
                         {p_{\dest{e}}}
                    \,\bigg\arrowvert\, e \in F, |\src{e}|=2 \bigg\},\\
      \varphi&:=p_t,
    \end{align*}
    and $\{\mathrm{src}_1(e), \mathrm{src}_2(e)\}$ denote the source nodes of $e$.
    As for the correctness, observe that if for some $e \in F$ with $|\src{e}|=2$
    both variables $p_{\mathrm{src}_1(e)}$ and $p_{\mathrm{src}_2(e)}$ can be derived from the 
    stable extension of $\encoding{W,D}$, then $p_{e}$ and consequently $p_{\dest{e}}$ 
    can be derived. Conversely, if $\mathrm{src}_1(e)$ or $\mathrm{src}_2(e)$ cannot 
    be derived, then either none or two of the propositions in 
    $p_{\mathrm{src}_1(e)} \xor p_{\mathrm{src}_2(e)} \xor p_{e}$ are satisfied.
    Thus $p_{\dest{e}}$ cannot be derived from the defaults corresponding to $e$.

    Finally, it remains to show $\NL$-hardness for $\CloneI_2 \subseteq [B]$.
    We give a $\leqcd$-reduction from $\GAP$ to $\CRED(\{\id\})$. 
    For a directed graph $G=(V,F)$ and two nodes $s,t\in V$, we
    transform the $\GAP$-instance $(G,s,t)$ with $G=(V,F)$ to the $\CRED(\CloneI_2)$-instance
    \[
      W       := \{p_s\}, \ 
      D       := \bigg\{ \frac{p_u : p_u}{ p_v } \,\bigg\arrowvert\,
                         (u,v) \in F \bigg\}, \ 
      \varphi:=p_t.
    \]
    Clearly, $(G,s,t) \in \GAP$ if and only if $\varphi$ is contained in all stable extensions of $\encoding{W,D}$.
  \end{proof}

  This completes the proof of Theorem \ref{thm:credulous-reasoning}.

  We will next classify the complexity of the skeptical reasoning
  problem. The analysis as well as the result are similar to the classification of
  the credulous reasoning problem (cf.\ also Fig.~\ref{fig:credulous-skeptical-reasoning}).

    \begin{theorem} \label{thm:skeptical-reasoning}
    Let $B$ be a finite set of Boolean functions. Then $\SKEP(B)$ is 
    \begin{enumerate}
      \item \label{thm:skeptical-reasoning-PiP2}
        $\PiPtwo$-complete if $\CloneS_{1} \subseteq [B] \subseteq
        \CloneBF$ or $\CloneD \subseteq [B] \subseteq \CloneBF$, 
      \item \label{thm:skeptical-reasoning-DeltaP2}
      $\DeltaPtwo$-complete if $\CloneS_{11} \subseteq [B] \subseteq \CloneM$, 
      \item \label{thm:skeptical-reasoning-coNP}
        $\co\NP$-complete if $X \subseteq [B] \subseteq Y$, where 
        $X \in \{\CloneS_{00}, \CloneS_{10}, \CloneN_2,\CloneL_0\}$
        and $Y \in \{\CloneR_1,\CloneM, \CloneL \}$,  
      \item \label{thm:skeptical-reasoning-P}
        $\P$-complete if $\CloneV_2 \subseteq [B] \subseteq \CloneV$,
        $\CloneE_2 \subseteq  [B] \subseteq \CloneE$ or
        $[B]\in\{\CloneL_1,\CloneL_2\}$, 
        and
      \item \label{thm:skeptical-reasoning-NL}
        $\NL$-complete if $\CloneI_2 \subseteq [B] \subseteq \CloneI$. 
    \end{enumerate}
  \end{theorem}
  
  \begin{proof}
    The first part again follows from
    Theorem~\ref{thm:gottlob} and Lemma~\ref{lem:cloneB-equivcd-cloneB+1}.
    
    For $[B] \in \{\CloneN,\CloneN_2,\CloneL,\CloneL_0,\CloneL_3\}$,
    we guess similarly as in Theorem~\ref{thm:extension-existence} 
    a set $G$ of defaults and then verify in the same way whether $W$
    and $G$ generate a stable extension $E$. If not, then we
    accept. Otherwise, we check if $E\models\varphi$ and answer according
    to this test. 
    This yields a  $\co\NP$-algorithm for $\SKEP(B)$.
    Hardness for $\co\NP$ is achieved by modifying the reduction from
    Theorem~\ref{thm:extension-existence} (cf.\ also the proof of
    Proposition~\ref{prop:credulous-reasoning-lower-bounds}): map 
    $\varphi$ to $(\encoding{\emptyset,D_\varphi},\psi)$, where
    $D_\varphi$ is defined as in the proof of
    Theorem~\ref{thm:extension-existence}, and $\psi$ is a 
    $B$-formula such that $\varphi$ and
    $\psi$ do not share variables. 
    Then $\varphi \notin \ThreeSAT$
    if and only if $\encoding{\emptyset,D_\varphi}$ does not have a stable
    extension.
    The latter is true if and only if $\psi$ is in all extensions
    of $\encoding{\emptyset,D_\varphi}$.
    Hence $\overline{\ThreeSAT}\leqcd \SKEP(B)$, establishing the claim.

    For all remaining clones $B$, observe that $[B] \subseteq
    \CloneR_1$ or $[B] \subseteq \CloneM$. Hence,
    Corollary~\ref{cor:skep(b)-equiv-cred(b)} and
    Theorem~\ref{thm:credulous-reasoning} imply the claim.   
  \end{proof}

\section{Conclusion} \label{sec:conclusion}

  In this paper we provided a complete classification of the
  complexity of the main reasoning problems for default propositional
  logic, one of the most common frameworks for nonmonotonic reasoning.
  The complexity of the extension existence problem shows an
  interesting similarity to the complexity of the satisfiability
  problem \cite{lew79}, because in both cases the 
  hardest instances lie above the clone $\CloneS_1$ (with the
  exception that instances from $\CloneD$ are still hard for
  $\EXT$, but easy for $\SAT$).
  The complexity of the membership problems, \emph{i.e.}, credulous and skeptical
  reasoning, rests on two sources: first, whether there exists a
  unique extension (cf.\ Lemma~\ref{lem:extension-R1-M-unique}), and
  second, how hard it is to test for formula implication. 
  For this reason,  we also classified the complexity of the
  implication problem $\IMP(B)$.

A different complexity classification of reasoning for default logic has been undertaken in \cite{chhesc07}. In that paper, the language of  existentially quantified propositional logic was restricted to so called conjunctive queries, \emph{i.e.}, existentially quantified formulae in conjunctive normal-form with generalized clauses. The complexity of the reasoning tasks was determined depending on the type of clauses that are allowed. We want to remark that though our approach at first sight seems to be more general (since we do not restrict our formulae to CNF), the results in \cite{chhesc07} do not follow from the results presented here (and vice versa, our results do not follow from theirs).
  
  In the light of our present contribution, it is interesting to remark that by
  results of Konolige, Gottlob, and Janhunen \cite{konolige:88,gottlob:95,janhunen:99}, propositional default logic and 
  Moore's autoepistemic logic are essentially equivalent.
  Even more, the translations are efficiently
  computable. Unfortunately, all of them require a complete
  set of Boolean connectives, whence our results do not immediately 
  transfer to autoepistemic logic. It is nevertheless interesting to
  ask whether the exchange of default rules with the introspective
  operator $L$ yields hitherto unclassified fragments of autoepistemic logic 
  that allow for efficient stable expansion testing and reasoning.

\subsection*{Acknowledgements}
  We thank Ilka Schnoor for sending us a manuscript with the proof of
  the $\NP$-hardness of $\EXT(B)$ for all $B$ such that $\CloneN_2 \subseteq [B]$. 
  We also acknowledge helpful discussions
  on various topics of this paper with Peter Lohmann.

\bibliographystyle{alpha}
\bibliography{thi-hannover}

\end{document}